\documentclass[aps,prl,twocolumn,showpacs,superscriptaddress,groupedaddress]{revtex4}

\usepackage{graphicx}
\usepackage{latexsym}
\usepackage{revsymb}
\usepackage{amsfonts}
\usepackage{amsmath}
\usepackage{amssymb}
\usepackage{bm}
\usepackage{amsthm}
\usepackage{color}
\usepackage{enumerate}

\newtheorem{theorem}{Theorem}
\newtheorem{lemma}{Lemma}

 {\begin{list}{}%
         {\setlength{\leftmargin}{#1}}%
         \item[]%
 }
 {\end{list}}

\newcommand{\bra}[1]{\langle #1 |}
\newcommand{\ket}[1]{| #1 \rangle}

\begin{document}

\title{Best Separable Approximation of multipartite diagonal symmetric states}

\author{Ruben Quesada}
\affiliation{Departament de F\'isica. Universitat Aut\`onoma de Barcelona 08193 Bellaterra, Spain}
\author{Anna Sanpera}
\affiliation{ICREA, Instituci\'o Catalana de Recerca i Estudis Avan\c cats, Barcelona, Spain}
\affiliation{Departament de F\'isica. Universitat Aut\`onoma de Barcelona 08193 Bellaterra, Spain}

\begin{abstract}
The structural study of entanglement in multipartite systems is hindered by the lack of necessary and sufficient operational criteria able to discriminate among the various entanglement properties of a given mixed state.  Here, we pursue a different route to the study of multipartite entanglement based on the closeness of a multipartite state to the set of separable ones. In particular, we analyze multipartite diagonal symmetric $N$ qubit states and provide the analytical expression for their Best Separable Approximation (BSA [Phys. Rev. Lett. {\bf 80}, 2261 (1998)]), that is, their unique convex decomposition into a separable part and an entangled one with maximal weight of the separable one.
\end{abstract}


\keywords{} \pacs{03.67.Mn, 03.67.-a, 03.65.Ud}

\maketitle


\section{Introduction}
Despite the spectacular success in developing, applying and even anticipating applications of quantum correlations, the study of multipartite entanglement remains an odd subject that often challenges our current understanding on the subject. While the structure of quantum states in the bipartite scenario is  well understood and several measures of entanglement have been firmly established, the multipartite setting remains still quite unexplored \cite{Horodecki09, Huber13, Guhne13}.

One very useful approach in the study of bipartite entanglement  comes from the so-called  Best Separable Approximation (BSA) of a density matrix \cite{Lewenstein98}. The BSA approximation of any state $\rho$ acting on a bipartite system ${\cal{H}}=\mathbb{C}^{d_{1}}\otimes \mathbb{C}^{d_{2}}$ is its convex decomposition as $\rho=\lambda\rho_S+(1-\lambda){\rho}_E$,  where $\rho_S$ is a separable density matrix, $\rho_E$ an entangled one and the weight $\lambda$ of the separable part is maximal. 
The state $\rho_S$ is then called the BSA of $\rho$. It is important to remark that although there exist many different decompositions of the form given above, the BSA is always unique, i.e. there exist a unique decomposition with maximal $\lambda$ \cite{Karnas01}. This uniqueness has been exploited to demonstrated important separability criteria in $2\times N$ systems as well as $N\times M$ \cite{Kraus00}. 
Further, it has been shown that when $\rho$ acts in $\mathbb{C}^{2}\otimes \mathbb{C}^{2}$, the BSA of any generic entangled state $\rho$ always leads to $\rho_E=\ket{\psi}\bra{\psi}$, that is, to a single projector into a pure entangled state  \cite{Lewenstein98}. 
Thus, for a qubit pair, the BSA provides trivially a measure of entanglement which can be expressed as $R(\rho)=(1-\lambda)S(\psi)$, where $S(\psi)$ measures the entanglement of $\ket{\psi}$ via the Von Neumann entropy of any of its subsystems. Later on, a remarkable connection between the BSA and the concurrence $C(\rho)$ was presented in \cite{Wellens01} namely, if $\rho$ has full rank (i.e. rank
4), then $\ket{\psi}$ is maximally entangled, and the BSA measure coincides with that of the concurrence i.e. $R(\rho) = C(\rho)$.  In this case, the BSA generates a pure-state decomposition of $\rho$ that minimizes the average concurrence, since the decomposition consists of a single pure state with $C (\psi)= 1$ and weight $(1-\lambda)$, together with other states all having $C(\psi) = 0$ \cite{Wootters02}. More recently, a connection between the BSA and the  max- relative entropy  defined as $D_{\max}(\tau || \sigma) \equiv \log \min \{\lambda : \tau \leq \lambda \sigma\}$, for $\tau, \sigma \geq 0$ and $\text{supp}\, \tau 	\subseteq \text{supp}\, \sigma $  has been found \cite{Datta13private}.
The max- relative entropy, $D_{max}(\rho_S||\rho)$, can be interpreted as the maximal probability in which a state $\rho_S$ appears in the convex decomposition of $\rho$. Clearly, maximizing over the set of separable states 
is by definition the {\it separability} $\lambda$ of $\rho$, i.e., $\lambda=\max (2^{-D_{max}(\rho_S||\rho)})$. The max- relative entropy plays the role of a parent entropy for some relevant quantum relative entropies \cite{Datta09,Renner09}.\\
Our aim here is to study if there exist a BSA in the multipartite case. That is, given a state $\rho$ acting on $\mathbb{C}^{d_{1}}\otimes \mathbb{C}^{d_{2}}\otimes\cdots\otimes\mathbb{C}^{d_{n}}$  we inquire if there exist a unique decomposition $\rho=\lambda{\rho}_S+(1-\lambda){\rho}_E$ with maximal $\lambda$. Notice that generically, such a decomposition cannot be determined due to the presence of entangled states that are positive under partial transposition (PPT entangled)\cite{Peres96,Horodecki96}. Here we focus on $N$-qubit symmetric states, whose separability conditions are less stringent due to the underlying symmetry. We demonstrate that in some cases it is possible to derive analitically the BSA for an arbitrary number of subsystems $N$.\\
The paper is organized as follows: in Section II we first review the notion of Best Separable Approximation to a given state. In Section III, we introduce  the subspace of symmetric N-qubit states, in particular the diagonal symmetric. In Section IV, we prove constructively the existence of the BSA for N-qubit diagonal symmetric states assuming the state is full NPT. We provide also the explicit expressions for the first non trivial case, namely $\mathbb{C}^{2}\otimes \mathbb{C}^{2}\otimes \mathbb{C}^{2}$. In Section V we conclude and list the open questions.

\section{Best Separable Approximation (BSA)} 
The BSA approach of a given state $\rho$ is based on subtracting maximally projectors onto product states contained on its range such that the remainder remains positively defined. Formally that can be expressed as: given a bipartite quantum system $\rho$ acting in a Hilbert space, $\cal{H}$ of finite dimension, and a set $V$ of product vectors contained on its range, i.e. $\ket{e_\alpha,f_\alpha} \in \cal{R}(\rho)$, there exists a separable state $\tilde{\rho}_S = \sum \Lambda_\alpha P_\alpha$, called the optimal separable approximation (OSA) of $\rho$, and a decomposition $\rho =\tilde{\rho}_S+\tilde{\rho}_E = \lambda {\rho}_S + (1-\lambda) {\rho}_E$, where $\rho_E$ does not contain any product vector on its range and $\lambda= Tr(\tilde{\rho}_S)$, where the tilde over a density matrix means it is not normalized. If the set $V$ includes {\it all product vectors} in $\mathcal{R}(\rho)$ then $\rho_S$ is called the best separable approximation (BSA) of $\rho$, and it is {\it unique} for every density matrix $\rho$ {\cite{Karnas01}}.
The conditions that a separable state $\rho_S$ has to fulfill in order to be the BSA of a given entangled state $\rho$ were given and proved in \cite{Lewenstein98}. They reduce to find the conditions for subtracting maximally (i) a projector $P$ onto a product state $\ket{p} \in \mathcal{R}(\rho)$ as well as (ii) a pair of projectors $P_1$, $P_2$ onto product vectors $\ket{p_1}$, $\ket{p_2} \in \mathcal{R}(\rho)$. For the easiness of our results we have included these maximality conditions in Appendix \ref{App:AppendixA} of the manuscript.

\section{Symmetric States}
General $N$-qubit symmetric states $\rho$ are defined as states lying in a subspace $\mathcal{S}_N$ of the Hilbert space $\mathcal{H}_{2,N} = (\mathbb{C}^2)^{ \otimes N}$ and fulfilling the relationship
\begin{equation}
\label{permutation}
V_{\sigma}\rho=\rho V_{\sigma'}^\dag = \rho
\end{equation}
for any permutations $\sigma,\sigma'$ over the $N$-element set, and where the operator $V_\sigma$ is defined as $V_\sigma \ket{\psi_1}\cdots\ket{\psi_N} = \ket{\psi_{\sigma(1)}}\cdots \ket{\psi_{\sigma(N)}}$, for $\ket{\psi_k} \in \mathbb{C}^2$ \cite{Tura12}. The normalized vectors that span $\mathcal{S}_N$ are
\begin{equation}
\label{N-basis}
\ket {D _k^N} = \frac{1}{{\sqrt {{C_k^N}} }}\sum\limits_{\sigma} {\ket {{1^k}{0^{N - k}}}} \,\,\,\,\,\,(k=0,\dots,N),
\end{equation}
where $1^k0^{N-k}$ denotes $1\mathop{\cdots} \limits^k 1 0\mathop{  \cdots} \limits^{N - k} 0$. The sum in (\ref{N-basis}) is carried over every different permutation $\sigma$ and ${C_k^N} \equiv \binom{N}{k}$. The dimension of $\mathcal{S}_N$ is $N+1$ we have that $\mathcal{S}_N \cong \mathbb{C}^{N+1}$ and thus for every bipartition $S|\bar{S}$ we have $\mathcal{S}_N \cong \mathbb{C}^{|S|+1}\otimes \mathbb{C}^{N-|S|+1}$.
The elements of the basis, also known as Dicke states, have many applications for quantum information processing \cite{Xing12,Prevedel09}.
If we consider the partition with $|S|=1$, it is straightforward to see that $S_N \simeq \mathbb{C}^2\otimes \mathbb{C}^N$, and our multipartite symmetric state can be thought as a bipartite one. Within this partition, the elements of the basis can be rewritten as
\begin{equation*}
\ket{D_k^N} = \frac{1}{\sqrt{C_k^N}}(\ket{0}\sqrt{C_k^{N-1}}\ket{D_k^{N-1}} + \ket{1}\sqrt{C_{k-1}^{N-1}}\ket{D_{k-1}^{N-1}}),
\end{equation*}
which using the substitution $\ket{\overline{k}} = \ket{D_k^{N-1}}$ can be readily expressed as
\begin{equation}
\label{2xNBasis}
\ket{D_k^{N}} = \sqrt{\frac{N-k}{N}}\ket{0}\ket{\overline k}+\sqrt{\frac{k}{N}}\ket{1}\ket{\overline {k-1}},
\end{equation}
for $k=0,\dots,N$.\\
In this paper we will focus on a particular subset of the symmetric states: the diagonal symmetric.  Let us denote by $\mathcal{P}_N$ the projector space onto $\mathcal{S}_N$, which is spanned by $\{\ket{D_k^N}\bra{D_k^N}\}$, for $k=0,\dots,N$. Hence, a general mixed state ${\rho}\in \mathcal{P}_N$, called Diagonal Symmetric State (DSS), is defined as
\begin{equation}
\label{2xN}
{\rho} = \sum\limits_{k = 0}^{N}{p_k\ket{D_k^N}\bra{D_k^N}},
\end{equation}
where $p_k \in [0,1]$ denote probabilities and $\sum p_k = 1$.\\
To illustrate the above description let us consider the case $N=3$ as it is the first non-trivial case and the last one in which the Peres-Horodecki (PH) criterion is necessary and sufficient. From (\ref{N-basis}) we see that $\{\ket{D_k^N}\}$ is given by
\begin{align}
\ket{D_0^3} &=\ket{000}\\
\ket{D _1^3} &= \frac{1}{{\sqrt 3 }}(\ket{001} + \ket{010} + \ket{100})\\
\ket{D _2^3} &= \frac{1}{{\sqrt 3 }}(\ket{011} + \ket{101} + \ket{110})\\
\ket{D _3^3} &=\ket{111}.
\end{align}

Using the identifications $\ket{\bar{0}} \equiv \ket{00}$, $\ket{\bar{1}} \equiv \frac{1}{\sqrt{2}}(\ket{01}+\ket{10})$ and $\ket{\bar{2}} \equiv \ket{11}$, we obtain a redefinition of the above basis as a bipartite $2\otimes 3$ system:
\begin{align}
\ket{D_0^3} &=\ket{0\bar{0}}\\
\ket{D _1^3} &= \frac{1}{{\sqrt 3 }}(\sqrt{2}\ket{0\bar{1}} + \ket{1\bar{0}})\\
\ket{D _2^3} &= \frac{1}{{\sqrt 3 }}(\ket{0\bar{2}} +\sqrt{2} \ket{1\bar{1}})\\
\ket{D _3^3} &=\ket{1\bar{2}}.
\end{align}
With the basis defined above for the $2\otimes 3$ case, a generic density matrix in the diagonal symmetric subspace takes the form:
\begin{equation}
\label{2x3Matrix}
\rho = \left[ {\begin{array}{*{20}{c}}
p_0 & 0 & 0 & 0 & 0 & 0\\
0    &  \frac{2}{3}p_1 & 0 & \frac{\sqrt{2}}{3}p_1 & 0 & 0\\
0 & 0 & \frac{1}{3}p_2 & 0 & \frac{\sqrt{2}}{3}p_2 & 0\\
0 & \frac{\sqrt{2}}{3}p_1 & 0 & \frac{1}{3}p_1 & 0 & 0\\
0 & 0 & \frac{\sqrt{2}}{3}p_2 & 0 & \frac{2}{3}p_2 & 0\\
0 & 0 & 0 & 0 & 0 & p_3
\end{array}} \right].
\end{equation}

\section{BSA of Diagonal Symmetric States}

The characterization of entanglement in $\mathcal{S}_N$ has been extensively addressed in recent years. Despite substantial progress (see \cite{Tura12,Augusiak12} and references therein)  full characterization of entanglement even in this simple case has  still not been achieved. The difficulty arises from the existence of PPT entangled states that are shown to exist for all N-qubit symmetric spaces with $N\geq 4$ \cite{Tura12}.\\
It is also known that if $\rho$ is PPT,  supported on $\mathbb{C}^2\otimes \mathbb{C}^N$ and has rank  $r(\rho)=N$, then the state is separable \cite{Kraus00}. Here we assume the generic case in which all probabilities $p_{k}\neq 0$, so that  $\rho$ has full rank, i.e. $r(\rho)=N+1$.
Moreover, in order to ensure the state is entangled we impose that $\rho^{T_B}$ ($\rho^{T_A}$) is not definite positive. Expressing $\rho$ in $\mathbb{C}^2 \otimes \mathbb{C}^N$ the conditions under which it happens are given by a set of inequalities
\begin{equation}
\label{NPTConds}
k(N-k)p_k^2 > (k+1)(N-k+1)p_{k-1}p_{k+1},
\end{equation}
for $k=1,\dots,N-1$. \\
The fulfillment of each of them imposes one negative eigenvalue for $\rho^{T_B}$ up to $N-1$ values, which is the maximum number of negative eigenvalues the partial transpose of a $2\times N$ state can have \cite{Rana13}. Whenever condition (\ref{NPTConds})  is fulfilled for every $k$, it can be rewritten as a different set of conditions akin to the first one (see Appendix B for the proof):
\begin{equation}
\label{Entang_conds}
p_k>C^N_k p_0^{\frac{N-k}{N}}p_N^{\frac{k}{N}},
\end{equation}
for $k=1,\dots,N-1$.

To proceed further we look for generic product vectors in the range of $\rho$. It is quite straightforward  to see that a generic product vector $\ket{p} = (X,e^{i\phi}Y)^{\otimes N} \in \mathcal{S}_N$, when expressed in $\mathbb{C}^2 \otimes \mathbb{C}^N$ reads	
\begin{eqnarray}
\ket{p} &=&  \left[ {\begin{array}{*{20}{c}}
  X \\ 
  e^{i\phi}Y 
\end{array}} \right] \otimes \left[ {\begin{array}{*{20}{c}}
  {{X^{N - 1}}} \\ 
  {\sqrt {N - 1}e^{i\phi} {X^{N - 2}}Y} \\ 
   \vdots  \\ 
  {\sqrt {C_k^{N - 1}}e^{ik\phi} {X^{N - k - 1}}{Y^k}} \\ 
   \vdots  \\ 
  {{e^{i(N-1)\phi}Y^{N - 1}}} 
\end{array}} \right]\\
\label{prod_vec}
&=& \sum\limits_{k = 0}^N {{e^{i\phi (N - k)}}\sqrt {C_k^N} {{X}^{N - k}}{Y^k}\ket{D _k^N}},
\end{eqnarray}
where $\phi, X$ and $Y$ are real parameters, and where the normalization can be included into the $X$ and $Y$. Among the infinite set of product vectors (\ref{prod_vec}), we choose those with the particular value $\phi=\frac{m\pi}{N}$:
\begin{equation}
\label{prod_vec2}
\ket{p_m} = \sum\limits_{k = 0}^N {\omega^{m(N-k)}\sqrt {C_k^N} {{X}^{N - k}}{Y^k}\ket{D _k^N}},
\end{equation}
for $m=0,\dots,2N-1$, where $\omega = e^{\frac{i\pi}{N}}$ is the primitive $2N$-root of unity ($\omega^{2N}=1$).
Product vectors (\ref{prod_vec2}) can be used to build the following separable (by construction) state
\begin{equation}
\label{rho_sep}
{\rho _S (X,Y)} \equiv \sum\limits_{m=0}^{2N - 1} {{\Lambda _m}{P_m}},
\end{equation}
where $P_m=\ket{p_m}\bra{p_m}$ are the projectors onto product vectors (\ref{prod_vec2}), and the  weights $\Lambda_m$ must fulfill $\sum \Lambda_m = 1$ and $\Lambda_m \geq 0$.

How much a generic product vector $\ket{p}=\sum {b_k \ket{D_k^N}} \in \mathcal{S}_N$ can be subtracted from the state $\rho = \sum {p_k \ket{D_k^N} \bra{D_k^N}}$ is given by equation (\ref{lemma1}) and reads
\begin{equation}
\label{LamdaEqual}
\Lambda=\frac{1} {\bra{p} \rho^{-1} \ket{p} } = \left(  \sum_k{ \frac{ |b_k|^2 } {p_k} }\right)^{-1},
\end{equation}
where we made use of the pseudoinverse
\begin{equation}
\rho^{-1} = \sum p_k^{-1} \ket{D_k^N}\bra{D_k^N}.
\end{equation}
In particular, since expression (\ref{LamdaEqual}) does not depend on the phase $\phi$  and  therefore, neither on $m$, it can be inferred that the $\Lambda_m$'s are the same for every $m$.

The normalization constraint $\sum_{m=1}^{2N} \Lambda_m = 1$ implies that $\Lambda_m=\frac{1}{2N}$, for all $m$. This fact, along with the relationship for $2N$-roots of unity, $\sum_{m=0}^{2N-1} \omega^{m(i-j)} = 2N\delta_{ij}$, shows that (\ref{rho_sep}) is diagonal in the $\{ \ket{D_k}\}$ basis, taking the following form
\begin{equation}
\label{RhoS}
\tilde{\rho}_S(X,Y)=\sum \limits_{k=0}^{N} C_k^N X^{2(N-k)}Y^{2k}\ket{D_k}\bra{D_k}.
\end{equation}
Additional constraints for $\tilde{\rho}_S$ to be the BSA of $\rho$ come by demanding $\rho-\tilde{\rho}_S \geq 0$, i.e. requiring $\tilde{\rho}_E \geq 0$ in the decomposition $\rho=\tilde{\rho}_S+\tilde{\rho}_E$, which leads to the conditions
\begin{equation}
\label{constraint}
p_k \geq C_k^N X^{2(N-k)}Y^{2k}\,\,\,\,\,\,\,(k=0,\dots,N).
\end{equation}

The similarity between constraints (\ref{constraint}) and conditions (\ref{Entang_conds}) leads naturally to the following choice
\begin{equation}
\begin{matrix}
\label{XY}
X^2 =p_0^{\frac{1}{N}} & & & & Y^2 = p_N^{\frac{1}{N}},
\end{matrix}
\end{equation}
which assures the positivity of the entangled part $\tilde{\rho}_E$ in the decomposition as far as the state $\rho$ is full NPT (\ref{Entang_conds}).
In what it follows we demonstrate that such a choice leads $\rho_S$ to be the BSA of $\rho$ and that it is valid only in the full NPT region of $\rho$. In other words we show that the region in which $\rho_S$ is the BSA of $\rho$ coincides with the region in which $\rho^{T_B}$ has the maximum number of negative eigenvalues. 

The main goal is now to prove that indeed for a generic $\rho \in \mathcal{S}_N$ of the form (\ref{2xN}), $\tilde{\rho}_S$ as constructed in (\ref{RhoS}) is its BSA under the assumption that $\rho$ is full NPT. According to \cite{Lewenstein98}, the weights $\Lambda_m$ of the projectors have to be maximal with respect to subtraction of one projector and subtraction of pairs of projectors. To this aim we prove first the following lemmas.
 
\begin{lemma}
\label{LemmaM}
Given $\rho_m \equiv \rho - \sum\limits_{j \ne m} \Lambda_j P_j = \rho-\tilde{\rho}_S+\frac{1}{2N}P_m$, then
\begin{equation}
\label{InverseMeq}
\bra{p_m}\rho_m^{-1}\ket{p_m} = 2N.
\end{equation}
\end{lemma}
\begin{proof}
By writting $\rho-\tilde{\rho}_S \equiv \sum q_k \ket{D_k}$ in the Dicke basis (\ref{2xNBasis}), the coefficients being $q_k = p_k - C_k^N p_0^{\frac{N-k}{N}}p_N^{\frac{k}{N}}$, for $k=0,\dots, N$. Then
\begin{equation}
\label{InverseM}
   \rho_m^{-1}	\ket{p_m}=\frac{N}{b_0^{(m)*}}\ket{D_0}+\frac{N}{b_N^{(m)*}}\ket{D_N} \equiv \ket{\chi_m},
\end{equation}
where the $b_k^{(m)}$ are the coefficients of $\ket{p_m}=\sum  b_k^{(m)} \ket{D_k}$.
To show this, it suffices to check that $\rho_m\ket{\chi_m} = \ket{p_m}$, by realizing that $P_m{\ket{\chi_m}}=2N\ket{p_m}$ and $(\rho-\tilde{\rho}_S)\ket{\chi_m} = 0$ (due to $q_0=q_N=0$). After  that, the lemma follows by applying $\bra{p_m} = \sum b_k^{(m)*} \bra{D_k}$ on the left side of (\ref{InverseM}).
\end{proof}

The following Lemma 2 refers to maximality with respect to a pair of projectors $P_{m,l}$. It will deal with the following state, for fixed values $m$ and $l$,
\begin{align}
\label{RhoML}
\rho_{ml} &\equiv\rho - \sum \limits_{j \ne m,l}\Lambda_j P_j\\
& = \rho - \rho_S + \frac{1}{2N}(P_m + P_l)\\
\label{RhoML2}
& = \sum q_k \ket{D_k}\bra{D_k}+\frac{1}{2N}(\ket{p_m}\bra{p_m} + \ket{p_l}\bra{p_l}),
\end{align}
where $\ket{p_m}=\sum b_k^{(m)} \ket{D_k}$, $\ket{p_l}=\sum b_k^{(l)} \ket{D_k}$ are product vectors defined as (\ref{prod_vec2}) and the relationship between their components is, again by virtue of (\ref{prod_vec2}),
\begin{equation}
\label{theta}
b_k^{(l)} = \omega^{(k-N)(m-l)} b_k^{(m)}\equiv \theta^{k-N} b_k^{(m)},
\end{equation}
where $\theta \equiv \omega 	^{m-l}=e^{\frac{i(m-l)\pi}{N}}$ has the property
\begin{equation}
\label{theta2}
\theta^N=\begin{cases}
+1 & \text{if $|m-l|$ is even}\\
-1 & \text{if $|m-l|$ is odd},
\end{cases}
\end{equation}
and trivially $|b_k^{(m)}|=|b_k^{(l)}|$.
Regarding this definition, we will now prove a lemma similar to the previous one.
\begin{lemma}
\label{InverseML}
For $\rho_{ml}$ defined in (\ref{RhoML}) we have that\\
\noindent (a)  if $|m-l|$ is odd then
\begin{equation}
\label{InverseMLa}
\bra{p_i}\rho_{ml}^{-1}\ket{p_j} =2N\delta_{ij},\,\,\,\,\,\,\,\,\,i,j\in \{m,l\}.
\end{equation}
(b) if $|m-l|$ is even then
\begin{equation}
\label{Lemma4b}
\bra{p_m}\rho_{ml}^{-1}\ket{p_m}+\bra{p_l}\rho_{ml}^{-1}\ket{p_m} =2N.
\end{equation}
\end{lemma}
\begin{proof}
(a) First of all we need to prove that $\rho_{ml}^{-1}\ket{p_m}=\frac{N}{b_0^{(m)*}}\ket{D_0}+\frac{N}{b_N^{(m)*}}\ket{D_N} \equiv \ket{\phi_m}$. To do that, it suffices to confirm that $\rho_{ml}\ket{\phi_m}=\ket{p_m}$:
\begin{align}
\rho_{ml} \ket{\phi_m}&= \frac{1}{2N} \left(\frac{N}{b_0^{(m)*}}b_0^{(m)*}+ \frac{N}{b_N^{(m)*}}b_N^{(m)*} \right) \ket{p_m}\\
&+\frac{1}{2N} \underbrace{\left(\frac{N}{b_0^{(m)*}}b_0^{(l)*}+ \frac{N}{b_N^{(m)*}}b_N^{(l)*} \right)}_{=0} \ket{p_l}=\ket{p_m},\notag
\end{align}
where the last expression in brackets vanishes because $b_0^{(l)*}=\theta^{N}b_0^{(m)*}=-b_0^{(m)*}$ (with $\theta^N=-1$ since $|m-l|$ is odd), while $b_N^{(l)*}=b_N^{(m)*}$. After that, (\ref{InverseMLa}) follows easily.

\noindent (b) In this case we first express $\rho_{ml}^{-1} \ket{p_m}$ in the $\{\ket{D_k} \}$ basis, without taking care of the value of its coefficients $a_k^{(m)}$:
\begin{align}
 \rho_{ml}^{-1} \ket{p_m} &= \sum a_k^{(m)} \ket{D_k} \equiv \ket{\psi_m}.
\end{align}
Now, using (\ref{RhoML2}) and (\ref{theta}) we obtain that, in order to fulfill the expression $\rho_{ml} \ket{\psi_m} = \ket{p_m}$, the following relationship must hold
\begin{equation}
\label{EqML}
\sum \limits_{i=0}^N a_i^{(m)} b_i^{(m)*}(1+\theta^{k-i}+\delta_{ki}NR_k)=2N,
\end{equation}
for $k=0,\dots,N,$ with the definition $R_k \equiv \frac{2q_k}{|b_k^{(m)}|^2}$, whenever $|b_k^{(m)}| \neq 0$.

Taking the $k=N$ component of (\ref{EqML})
\begin{equation}
\sum \limits_{i=0}^N a_i^{(m)} b_i^{(m)*}(1+\theta^{N-i})=2N,
\end{equation}
we obtain the desired result since
\begin{equation}
\bra{p_m}\rho_{ml}^{-1}\ket{p_m} =\sum \limits_{i=0}^N a_i^{(m)} b_i^{(m)*} 
\end{equation}
and
\begin{equation}
\bra{p_l}\rho_{ml}^{-1}\ket{p_m} =\sum \limits_{i=0}^N  a_i^{(m)} b_i^{(m)*}\theta^{N-i}.
\end{equation}
Note that due to the invariance $\rho_{ml}=\rho_{lm}$ we have the same expression (\ref{Lemma4b}) with $m$ and $l$ exchanged. Also remark that $\bra{p_m}\rho_{ml}^{-1}\ket{p_l}=\bra{p_l}\rho_{ml}^{-1}\ket{p_m}$ due to hermiticity of $\rho_{ml}$.
\end{proof}

Now we have all the tools at hand to state the main theorem of this paper, whose proof, which is tedious but straightforward, has been moved to Appendix \ref{App:C}.

\begin{theorem}
\label{BSASymTheorem}
Given a full NPT $N$-qubit symmetric state of the form $\rho = \sum\nolimits_{k = 0}^{N}{p_k\ket{D_k^N}\bra{D_k^N}}$ where $\ket{D_k^N}$ are the normalized Dicke states of dimension $N$, then $\rho$ can be decomposed as ${\rho}= \lambda {\rho}_{BSA} + (1-\lambda) {\rho}_E$, where $\tilde{\rho}_{BSA}$ (not normalized) is given by
\begin{equation}
\label{BSA}
\tilde{\rho}_{BSA}=\tilde{\rho}_S\left(\sqrt[2N]{p_0},\sqrt[2N]{p_N}\right)=\frac{1}{2N}\sum \limits_{k=0}^{2N-1}P_k
\end{equation}
is the BSA of $\rho$ and $\lambda = Tr (\tilde{\rho}_{BSA} )= (\sqrt[N]{p_0}+\sqrt[N]{p_N})^N$ is the so-called {\it separability} of $\rho$. 
\end{theorem}

According to this theorem, provided that $$p_k \geq C^N_k p_0^{\frac{N-k}{N}}p_N^{\frac{k}{N}},$$ for $k=1,\dots,N-1$, the optimal and unique decomposition of $\rho$ reads
\begin{equation}
\label{BSAdec}
\rho= \lambda {\rho}_{BSA} + (1-\lambda)\rho_E,
\end{equation}
where $r(\rho)=r(\rho_{BSA})=N+1$, $\rho_E$ is a full NPT state of rank $N-1$ and the weight $\lambda$ has the maximum value that a separable state can achieve in decomposition (\ref{BSAdec}).

Let us again illustrate our findings with an explicit example for the simplest $N=3$ case (\ref{2x3Matrix}).
Using the PPT criterion we obtain that this particular state $\rho$ is entangled iff  either $p_2^2>3p_3p_1$ and/or  $p_1^2>3p_0p_2$ (\ref{NPTConds}). The full NPT conditions thus read
\begin{align}
p_{1} &> 3p_0^{2/3}p_3^{1/3}\\
p_{2} &> 3p_0^{1/3}p_3^{2/3},
\end{align}
in accordance with (\ref{Entang_conds}). Hence, a generic separable state $\rho_S$ of the form (\ref{2x3Matrix}) is given by $p_0$, $p_1'$, $p_2'$, $p_3$ with $p_1' \leq 3p_0^{2/3}p_3^{1/3}$ and $p_2' \leq 3p_0^{1/3}p_3^{2/3}$. Hence, the problem of finding the BSA of a given entangled $\rho$ (full NPT in this case) corresponds to the following maximization problem
\begin{equation}
\max \limits_{\rho_S \in \mathcal{S}}  \left\lbrace \mbox{Tr}(\tilde{\rho}_S) | \rho-\tilde{\rho}_S \geq 0 \right\rbrace,
\end{equation}
where $\mbox{Tr}(\tilde{\rho}_S) = p_0 + p_1'+p_2'+p_3$.
This problem, provided that $\rho$ is full NPT, has the simple solution
\begin{align}
p_1' &= 3p_0^{2/3}p_3^{1/3}\\
p_2' &= 3p_0^{1/3}p_3^{2/3},
\end{align}
and therefore the BSA of $\rho$ is given by 
\begin{equation*}
\label{2x3BSA}
\tilde{\rho}_{BSA}=\!\!\left[ {\begin{array}{*{20}{c}}
p_0 & 0 & 0 & 0 & 0 & 0\\
0    &  2p_0^{\frac{2}{3}}p_3^{\frac{1}{3}} & 0 & \sqrt{2}p_0^{\frac{2}{3}}p_3^{\frac{1}{3}} & 0 & 0\\
0 & 0 & p_0^{\frac{1}{3}}p_3^{\frac{2}{3}} & 0 & \sqrt{2}p_0^{\frac{1}{3}}p_3^{\frac{2}{3}}  & 0\\
0 & \sqrt{2}p_0^{\frac{2}{3}}p_3^{\frac{1}{3}} & 0 & p_0^{\frac{2}{3}}p_3^{\frac{1}{3}} & 0 & 0\\
0 & 0 & \sqrt{2}p_0^{\frac{1}{3}}p_3^{\frac{2}{3}}  & 0 & 2p_0^{\frac{1}{3}}p_3^{\frac{2}{3}}  & 0\\
0 & 0 & 0 & 0 & 0 & p_3
\end{array}} \right]
\end{equation*}
with 
\begin{eqnarray}
\lambda \equiv \mbox{Tr}\tilde{\rho}_{_{BSA}} &=& p_0 + 3p_0^{2/3}p_3^{1/3}+3p_0^{1/3}p_3^{2/3} + p_3 \nonumber\\
&=&(\sqrt[3]{p_0}+\sqrt[3]{p_3})^3
\end{eqnarray} 
in agreement with Theorem \ref{BSASymTheorem}, and consequently the state (\ref{2x3Matrix}) has the optimal decomposition $\rho = \lambda \rho_{BSA} + (1-\lambda) \rho_E$, with $\tilde{\rho}_E = \rho - \tilde{\rho}_{BSA}$, where, as before, the tilde indicates the state is not necessarily normalized.

In Figure \ref{fig:Region} we display the feasible region for the BSA of a full NPT $N$-qubit diagonal symmetric state as a function of $p_0$ and $p_N$ for different number of qubits $N$. For each $N$, the region covers the  shaded area below the curve given by the boundary of $\lambda = (\sqrt[N]{p_0}+\sqrt[N]{p_N})^N \leq 1$. The dashed region corresponds to unphysical states and the area between dashed and shaded regions corresponds to non full NPT entangled states, including possible bound states. As it can be seen, increasing the number of parties $N$, the region of full NPT shrinks dramatically. It is also important to notice that whenever $\rho$ is full NPT the parameter $\lambda$ is upper bounded by $1$:
\begin{eqnarray}
 \lambda &=& (\sqrt[N]{p_0}+\sqrt[N]{p_N})^N \\
 &=& \sum \limits_{k=0}^N {C_k^N p_0^{\frac{N-k}{N}} p_N^{\frac{k}{N}}}\\
 &<& \sum \limits_{k=0}^N p_k = 1, 
\end{eqnarray}
where the last inequality stands by virtue of (\ref{Entang_conds}).

\begin{figure}
  \centering
    \includegraphics[width=0.45\textwidth]{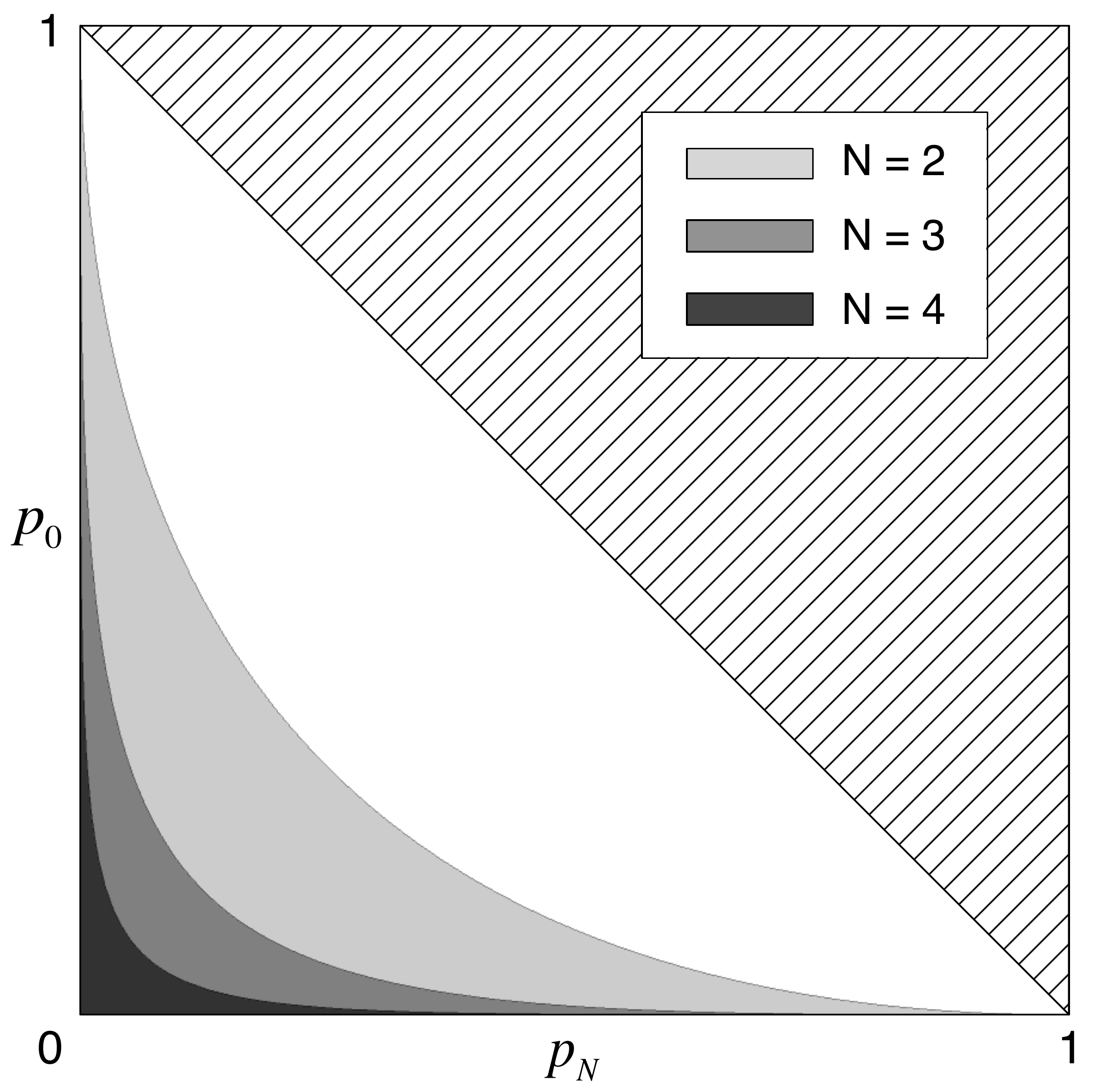}
  \caption{Feasible region for the BSA of full NPT $N$-qubit symmetric states as a function of $p_0$ and $p_N$, for $N=2, 3$ and $4$. For each $N$, the region covers the  shaded area below the curve given by the boundary of $\lambda = (\sqrt[N]{p_0}+\sqrt[N]{p_N})^N \leq 1$. Dashed region corresponds to unphysical states and the area between dashed and shaded regions corresponds to non full NPT entangled states, for each $N$.}
  \label{fig:Region}
\end{figure}
\section{Conclusions}
We have provided an analytical closed form for the best separable approximation of  N-qubit Diagonal Symmetric  States (DSS) and in turn this lead to an analytical measure of entanglement (separability) for those multipartite states, given by the coefficient $\lambda \in [0,1]$, which for $\lambda = 1$ it is separable and the value $\lambda = 0$ is out of our scope, since in that case $\rho$ is no more full NPT.

The simplicity of the construction relies on the fact that (i) states considered are diagonal in the Dicke basis and (ii)  there are $N-1$ entangled states in the range of $\rho$ (full NPT).     
It is interesting to notice that altough $r(\rho_{BSA} )= N+1$, the BSA depends only on the  probabilities $p_0$ and $p_N$, corresponding to the two unique product  states on the Dicke basis ($\ket{0 \cdots 0}$ and $\ket{1 \cdots 1}$). Numerically we have seen that this feature is far from being general and, in fact, it does not hold outside the range of full NPT of $\rho$. That is, there exists entangled states whose BSA depend on the in-between probabilities $p_1,\dots, p_{N-1}$. It would remain to consider the situation where $\rho$ is non-full NPT (i.e. with $1 \leq  t <  N-1$, where $t$ stands for the number of negative eigenvalues of $\rho^{T_B}$) in order to complete the separability problem of DSS.

\section*{ACKNOWLEDGMENTS}
We thank N. Datta for discussions. We acknowledge financial support from the
Spanish MINECO (FIS2008-01236), European Regional development Fund, Generalitat de
Catalunya Grant No. SGR2009-00347. RQ acknowledges Spanish MECD for the FPU Fellowship FPU12/03323.

\appendix
\section{} \label{App:AppendixA}
\label{App:A}
Here we present the conditions that $\rho_S$ need to fulfill in order to be the BSA of a given density matrix $\rho$, as shown in \cite{Lewenstein98}.

{\bf \noindent Definition 1.} A non-negative parameter $\Lambda \in \mathbb{R}$ is called {\it maximal} with respect to a density matrix $\rho$, and the projection operator $P=\ket{\psi}\bra{\psi}$ iff $\rho-\Lambda P \geq 0$, and for every $\epsilon \geq 0$, the matrix $\rho -(\Lambda + \epsilon)P$ is not positive definite.

\begin{lemma}[\cite{Lewenstein98}]
\label{maximal1}
$\Lambda$ is maximal with respect to $\rho$ and $P = \ket{\psi} \bra{\psi}$ iff
\begin{enumerate}[(a)]
\item  if $\ket{\psi} \notin \mathcal{R}(\rho)$ then $\Lambda = 0$,
\item if $\ket{\psi} \in \mathcal{R}(\rho)$ then \begin{equation}
\label{lemma1}
0 < \Lambda=\frac{1}{\bra{\psi}\rho^{-1}\ket{\psi}}.
\end{equation}
\end{enumerate}
\end{lemma}

{\bf \noindent Definition 2.} A pair of non-negative parameters $(\Lambda_1,\Lambda_2)$ is called {\it maximal} with respect to $\rho$ and a pair of projection operators $P_1 = \ket{\psi_1}\bra{\psi_1}$, $P_2 = \ket{\psi_2}\bra{\psi_2}$ iff $\rho - \Lambda_1 P_1 - \Lambda_2 P_2 \geq 0$, $\Lambda_1$ is maximal with respect to $\rho - \Lambda_2 P_2$, $\Lambda_2$ is maximal with respect to $\rho - \Lambda_1 P_1$, and the sum $\Lambda_1 + \Lambda_2$ is maximal. 

\begin{lemma}[\cite{Lewenstein98}]
\label{maximal2}
A pair $(\Lambda_1,\Lambda_2)$ is maximal with respect to $\rho$ and a pair of  projectors $(P_1,P_2)$ iff
\begin{enumerate}[(a)]
\item if $\ket{\psi_1},\ket{\psi_2} \notin \mathcal{R}(\rho)$ then $\Lambda_1=\Lambda_2=0$.
\item if $\ket{\psi_1}\notin \mathcal{R}(\rho)$, while $\ket{\psi_2}\in \mathcal{R}(\rho)$ then $\Lambda_1=0$, $\Lambda_2 = \bra{\psi_2}\rho^{-1}\ket{\psi_2}^{-1}$.
\item if $\ket{\psi_1},\ket{\psi_2} \in \mathcal{R}(\rho)$ and $\bra{\psi_1}\rho^{-1}\ket{\psi_2}=0$ then $\Lambda_i=\bra{\psi_i}\rho^{-1}\ket{\psi_i}^{-1}, i=1,2.$
\item if $\ket{\psi_1},\ket{\psi_2} \in \mathcal{R}(\rho)$ and $\bra{\psi_1}\rho^{-1}\ket{\psi_2} \neq 0$ then
\begin{align}
\label{lemma2}
\Lambda_1&=(\bra{\psi_2}\rho^{-1}\ket{\psi_2}-|\bra{\psi_1}\rho^{-1}\ket{\psi_2}|)/D\\
\Lambda_2&=(\bra{\psi_1}\rho^{-1}\ket{\psi_1}-|\bra{\psi_1}\rho^{-1}\ket{\psi_2}|)/D,
\end{align}
\end{enumerate}
where $D=\bra{\psi_1}\rho^{-1}\ket{\psi_1}\bra{\psi_2}\rho^{-1}\ket{\psi_2}-|\bra{\psi_1}\rho^{-1}\ket{\psi_2}|^{2}$.
\end{lemma}

\begin{theorem}[\cite{Lewenstein98}]
\label{BSATheorem}
Given the set $V$ of product  vectors $\ket{e,f} \in \mathcal{R}(\rho)$, the matrix $\rho_S=\sum\nolimits_k \Lambda_k P_k$ is the OSA of $\rho$ iff
\begin{enumerate}[(a)]
\item all $\Lambda_k$ are maximal with respect to $\rho_m = \rho-\sum \nolimits _{k \neq m} \Lambda_k P_k$, and to the projector $P_m$
\item all pairs $(\Lambda_m,\Lambda_l)$ are maximal with respect to $\rho_{m l}=\rho - \sum \nolimits _ {k \neq m,l}\Lambda_k P_k$, and to the projection operators $(P_m, P_l)$.
\end{enumerate}
\end{theorem}

\section{} \label{App:AppendixB}
\label{App:B}
Here we give the steps needed to convert the initial set of inequalities (\ref{NPTConds}) to the more convenient set (\ref{Entang_conds}).
First we need to prove the following relation
\begin{equation}
\label{1ineq}
{p_{N - j}} > {\left( {\frac{{j + 1}}{{N - j}}} \right)^{\frac{j}{{j + 1}}}}{\left( {C_{N - j}^N} \right)^{\frac{1}{{j + 1}}}}p_{N - (j + 1)}^{\frac{j}{{j + 1}}}p_N^{\frac{1}{{j + 1}}},
\end{equation}
where $C_k^N$ is the binomial coefficient.
To do that we will make use of the induction method. The expression with $j=1$ (i.e. $k=N-1$) can be obtained directly from (\ref{NPTConds}). It remains to prove the implication $j-1 \Rightarrow j$.
Taking as true the inequality (\ref{1ineq}) for $j-1$:
\begin{equation*}
{p_{N - (j - 1)}} > {\left( {\frac{j}{{N - (j - 1)}}} \right)^{\frac{{j - 1}}{j}}}{\left( {C_{N - (j - 1)}^N} \right)^{\frac{1}{j}}}p_{N - j}^{\frac{{j - 1}}{j}}p_N^{\frac{1}{j}}
\end{equation*}
and inserting it into the inequality (\ref{NPTConds}) with $k=N-j$,
\begin{equation*}
p_{N - j} > {\left[ {\frac{{(N - j + 1)(j + 1)}}{{j(N - j)}}} \right]^{\frac{1}{2}}}p_{N - (j + 1)}^{1/2}p_{N - (j - 1)}^{1/2},
\end{equation*} 
one obtains, after isolating $p_{N-j}$,
\begin{align*}
&{p_{N - j}} >\\
& {\left( {\frac{{j + 1}}{{N - j}}} \right)^{\frac{j}{{j + 1}}}}{\left[ {\frac{{N - (j - 1)}}{j}C_{N - (j + 1)}^N} \right]^{\frac{1}{{j + 1}}}}p_{N - (j + 1)}^{\frac{j}{{j + 1}}}p_N^{\frac{1}{{j + 1}}},
\end{align*}
which becomes (\ref{1ineq}) due to the combinatorial relationship
\begin{equation}
\label{Cn}
C_{N - j}^N = \frac{{N - (j - 1)}}{j}C_{N - (j - 1)}^N.
\end{equation}
Once we have proved this intermediate set of inequalities, we can focus on showing our goal inequalities 
\begin{equation}
\label{2ineq}
{p_k} > C_k^Np_0^{\frac{{N - k}}{N}}p_N^{\frac{k}{N}}.
\end{equation}
To this end, we will make use again of the induction method. The case $k=1$ easily comes from (\ref{1ineq}) with $j=N-1$. Now, supposing valid the case $k-1$ of (\ref{2ineq}):
\begin{equation}
{p_{k - 1}} > C_{k - 1}^Np_0^{\frac{{N - k + 1}}{N}}p_N^{\frac{{k - 1}}{N}},
\end{equation}
and inserting it into the inequality $j=N-1$ of (\ref{1ineq})
\begin{equation*}
{p_k} > {\left( {\frac{{N - k + 1}}{k}} \right)^{\frac{{N - k}}{{N - k + 1}}}}{\left( {C_k^N} \right)^{\frac{1}{{N - k + 1}}}}p_{k - 1}^{\frac{{N - k}}{{N - k + 1}}}p_N^{\frac{1}{{N - k + 1}}}
\end{equation*}
one obtains the desired set of inequalities (\ref{2ineq}) after making use of the combinatorial relationship (\ref{Cn}).

\section{} \label{App:AppendixC}
\label{App:C}

 \subsubsection*{Proof of Theorem \ref{BSASymTheorem} }

The procedure to prove this theorem is to check that the conditions given by Theorem \ref{BSATheorem} in Appendix \ref{App:A} are fulfilled, confirming that our ansatz $\Lambda_k=\frac{1}{2N}$, for all $k$ is indeed the maximal one.\\
First condition of Theorem \ref{BSASymTheorem} imposes that each $\Lambda_m$ must be maximal with respect to $\rho_m$ and $P_m$. This condition of maximality is given by  Lemma \ref{maximal1}. Hence we need to compute $\bra{p_m}\rho_m^{-1}\ket{p_m}$. It follows immediately from Lemma \ref{LemmaM} that, for every $m$,
\begin{equation}
\Lambda_m=\frac{1}{\bra{p_m}\rho_m^{-1}\ket{p_m}}=\frac{1}{2N}.
\end{equation}

\noindent Second condition of Theorem \ref{BSATheorem} imposes that all pair weights $(\Lambda_m,\Lambda_l)$ must be maximal with respect to $\rho_{ml}$ and the projectors $P_m$ and $P_l$.\\
Conditions for maximality of pairs $(\Lambda_m,\Lambda_l)$ are provided by Lemma \ref{maximal2}. We divide this section of the proof in two parts regarding the parity of $|m-l|$.
 
 \indent$\bullet$ When $|m-l|$ is odd:\\
We need to use part (a) of Lemma \ref{InverseML}. It assures that $\bra{p_l}\rho_{ml}^{-1}\ket{p_m} = 0$ for every $m,l$. Hence, to prove maximality of  pairs $(\Lambda_m,\Lambda_l)$ we need to focus on part (c) of Theorem \ref{BSATheorem}. This part establishes that maximal $\Lambda_i$ must be
\begin{equation}
\Lambda_i=\frac{1}{\bra{p_i}\rho_{ml}^{-1}\ket{p_i}},\,\,\,\,\,\,\,\,i=m,l.
\end{equation}
Once again, part (a) of Lemma \ref{InverseML} gives us immediately $\Lambda_i = \tfrac{1}{2N}$, for every $i$, in agreement with our ansatz.

 \indent$\bullet$ When $|m-l|$ is even:\\
In this case we have $\bra{p_l} \rho_{ml}^{-1}\ket{p_m} \neq 0$, so we need to follow part (d) of Theorem \ref{BSATheorem}. Accordingly, maximal $\Lambda_l$ must take the form 
\begin{equation}
\label{LambdaEven}
\Lambda_l=\frac{\bra{p_m}\rho_{ml}^{-1}\ket{p_m}-|\bra{p_m} \rho_{ml}^{-1}\ket{p_l}|}{\bra{p_m}\rho_{ml}^{-1}\ket{p_m}\bra{p_l}\rho_{ml}^{-1}\ket{p_l}-|\bra{p_m}\rho_{ml}^{-1}\ket{p_l}|^2}.
\end{equation}

\noindent Before dealing with (\ref{LambdaEven}) it is important to remark that $\bra{p_m}\rho_{ml}^{-1}\ket{p_l} \in \mathbb{R}$, which can be seen from part (b) of Lemma \ref{InverseML} since hermiticity of $\rho_{ml}$ imposes $\bra{p_m}\rho_{ml}^{-1}\ket{p_m} \in \mathbb{R}$. Also note that it is positive since $\rho_{ml}$ is positive semidefinite, so we can remove the complex modulus in (\ref{LambdaEven}).
Naming $A$ the numerator and $B$ the denominator of (\ref{LambdaEven}), we have that
\begin{align}
A &=\bra{p_m}\rho_{ml}^{-1}\ket{p_m}-\bra{p_m}\rho_{ml}^{-1}\ket{p_l} \\
&= 2N - 2\bra{p_m}\rho_{ml}^{-1}\ket{p_l}, 
\end{align}
and, on the other hand,
\begin{align}
B &= \bra{p_m}\rho_{ml}^{-1}\ket{p_m}\bra{p_l}\rho_{ml}^{-1}\ket{p_l}-\bra{p_m}\rho_{ml}^{-1}\ket{p_l}^2\\
&= \left(2N - \bra{p_m}\rho_{ml}^{-1}\ket{p_l}\right)\left(2N-\bra{p_m}\rho_{ml}^{-1}\ket{p_l}\right)\\
&-\bra{p_m}\rho_{ml}^{-1}\ket{p_l}^2 \\
& = 2N\left(2N- 2\bra{p_m}\rho_{ml}^{-1}\ket{p_l} \right) = 2NA,
\end{align}
 where we have used part (b) of Lemma \ref{InverseML}.
Finally we can conclude that $\Lambda_l = \frac{A}{B}= \frac{1}{2N}$, for all $l$, in agreement again with our ansatz.
This completes the proof of maximality. It remains to be shown that the OSA $\rho_{BSA}$ is, in addition, the BSA of $\rho$. To do that, it suffices to show that no more projectors $Q$ other than those included in $V= \{P_0,P_1,\dots,P_{2N-1} \}$ can be subtracted from $\tilde{\rho}_E = \rho - \tilde\rho_{BSA}$ maintaining the positivity of the difference, or put in other words, that the range of $\tilde{\rho}_E$ does not contain product vectors.
To do that, we construct a projector onto generic product vector $Q=\ket{p}\bra{p}$, where $\ket{p}$ is given by (\ref{prod_vec}). It can be shown that $\Delta\rho \equiv \tilde{\rho}_E - \epsilon Q$  has  negative diagonal entries $(\Delta \rho)_{00}=-\epsilon X^{2N}$ and/or $(\Delta \rho)_{NN}=-\epsilon Y^{2N}$ for all $\epsilon  > 0$ and for any value of $X$ and $Y$, and thus $Q$ cannot be subtracted from $\tilde \rho_E$.


\end{document}